\documentclass[a4paper,USenglish,cleveref, autoref, thm-restate]{lipics-v2021}
\hideLIPIcs

\title{Revisiting the Sparse Matrix Compression Problem}

\author{Vincent Jugé}{Université Gustave Eiffel}{}{}{}

\author{Dominik {Köppl}}{University of Yamanashi, Kofu, Japan \and \url{https://dkppl.de}}{dkppl@dkppl.de}{0000-0002-8721-4444}{}

\author{Vincent Limouzy}{University Clermont Auvergne}{}{}{}

\author{Andrea Marino}{University of Florence}{}{}{}

\author{Jannik Olblich}{University of Ulm}{}{https://orcid.org/0000-0003-3291-7342}{}

\author{Giulia Punzi}{University of Pisa}{giulia.punzi@unipi.it}{https://orcid.org/0000-0001-8738-1595}{}

\author{Takeaki Uno}{National Institute of Informatics}{}{}{}

\authorrunning{V. Jugé, D. Köppl, V. Limouzy, A. Marino, J. Olblich, G. Punzi, and T. Uno} 

\Copyright{Vincent Jugé, Dominik Köppl, Vincent Limouzy, Andrea Marino, Jannik Olblich, Giulia Punzi, and Takeaki Uno} 



\category{} 

\relatedversion{} 

\supplement{Source code for re-running the experiments via docker can be accessed at \url{https://github.com/koeppl/sparsematrix}.}

\acknowledgements{This work was supported by the Research Institute for Mathematical Sciences,
an International Joint Usage/Research Center located in Kyoto University,
and a research stay thanks to Bézout Labex.
We thank Jesper Jansson, Hideo Bannai, and the participants at LA Symposium 2024\&2025 Winter for constructive discussion.
The research was supported by ROIS NII Open Collaborative Research 2026 Grant Number 252M-23667,
and JSPS KAKENHI Grant Numbers JP23H04378 and JP25K21150. GP was supported by the Italian Ministry of Research, under the complementary actions to the NRRP “Fit4MedRob - Fit for Medical Robotics” Grant (\# PNC0000007).}

\EventEditors{John Q. Open and Joan R. Access}
\EventNoEds{2}
\EventLongTitle{42nd Conference on Very Important Topics (CVIT 2016)}
\EventShortTitle{CVIT 2016}
\EventAcronym{CVIT}
\EventYear{2016}
\EventDate{December 24--27, 2016}
\EventLocation{Little Whinging, United Kingdom}
\EventLogo{}
\SeriesVolume{42}
\ArticleNo{23}

\def\problembox#1{\vspace{2mm}\noindent\fbox{\begin{minipage}{0.985\linewidth}#1
    \end{minipage}}\vspace{2mm}}

\usepackage{algpseudocodex}
\usepackage{algorithm}

\usepackage[utf8]{inputenc}
\usepackage{xspace}
\usepackage{xcolor}
\usepackage{graphicx}
\usepackage{amsmath,amssymb,amsthm}
\usepackage[natbib,cref]{nikis}
\usepackage{listings}
\usepackage{diagbox}

\newcommand*{\probMaxShift}{\textsc{minMaxShift}\xspace}
\newcommand*{\probMaxShiftString}{\textsc{stringMinMaxShift}\xspace}
\newcommand*{\probMinLength}{\textsc{minLength}\xspace}
\newcommand*{\probMaxShiftTiles}{\textsc{tilesMinMaxShift}\xspace}

\begin{document}
\nolinenumbers

\maketitle

\begin{abstract}
The sparse matrix compression problem asks for a one-dimensional representation of a binary $n \times \ell$ matrix, formed by an integer array of row indices and a shift function for each row, such 
that accessing a matrix entry is possible in constant time by consulting this representation. 
It has been shown that the decision problem for finding an integer array of length $\ell+\rho$ or restricting the shift function up to values of $\rho$ is NP-complete (cf.~\ the textbook of Garey and Johnson). 
As a practical heuristic, a greedy algorithm has been proposed to shift the $i$-th row until it forms a solution with its predecessor rows.
Despite that this greedy algorithm is cherished for its good approximation in practice, we show that it actually exhibits an approximation ratio of $\Theta(\sqrt{\ell+\rho})$.
We give further hardness results for parameterizations such as the number of distinct rows or the maximum number of non-zero entries per row.
Finally, we devise an DP-algorithm that solves the problem for double-logarithmic matrix widths or logarithmic widths for further restrictions.
We study all these findings also under a new perspective by introducing a variant of the problem, where we wish to minimize the length of the resulting integer array by trimming the non-zero borders, which has not been studied in the literature before but has practical motivations.
\end{abstract}

\clearpage
\setcounter{page}{1}

\section{Introduction}
\label{sec:intro}

Binary matrices have ubiquitous applications in computer science, such as in databases, data mining, and machine learning.
For instance, a binary matrix can represent an incidence matrix of a bipartite graph or a transition matrix of a finite automaton.
In practice, these matrices are often sparse, which means that most of the entries are zero.
To save space and time, it is desirable to compress these matrices.
One way to compress a sparse matrix is to represent it as a one-dimensional array of row indices and a shift function for each row.
Such a representation allows for constant-time access to the matrix, and has already been proposed in the 1970s. It was also studied in the context of compilers and databases in the 1980s~\cite{aho86compilers}.
The problem of finding such a representation is known as the \emph{sparse matrix compression problem}, which we denote by \probMaxShift{}.
The original version of the decision problem of \probMaxShift{} is defined as follows:

\problembox{\probMaxShift{}, {\cite[Chapter A4.2, Problem SR13]{garey79computers}}
\\
	\textbf{Input:} 
An $n \times \ell$ binary matrix $M[1..n][1..\ell]$ with $n$ rows and $\ell$ columns and 
entries $M[i][j] \in \{ 0, 1 \}$ for all $i \in [1..n]$, $j \in [1..\ell]$, 
and an integer $\rho \in [0..\ell \cdot (n-1)]$.
\\
\textbf{Task:}
Decide whether there exists an integer array $P[1..\ell+\rho]$ with $P[i] \in [0..n]$ for every $i \in [1..\ell+\rho]$,
and a function $s: [1..n] \rightarrow [0..\rho]$ such that
$M[i][j] = 1 \Leftrightarrow P[s(i) + j] = i$ 
for all $i \in [1..n]$ and $j \in [1..\ell]$.
}
In what follows, we call $P$ the \emph{placement} and $s$ the \emph{shift function} of the representation of $M$ asked by \probMaxShift{}.
Moreover, we visualize $0$ by $\star$ and $1$ by $i$ in the $i$-th row of $M$ for better readability, for each $i \in [1..n]$.

\begin{example}
\label{example:smc}
Consider the \probMaxShift~problem with $\rho= 2$ for the following $3\times 5$ matrix $M$ 
on the left:

\begin{minipage}{0.45\textwidth}
    $$M = \begin{pmatrix}
        \star & 1 & \star & \star & 1 & \star \\
        \star & 2 & \star & 2 & \star & \star \\
        \star & 3 & \star & \star & \star & 3
\end{pmatrix} $$
\end{minipage}
\hfill
\begin{minipage}{0.45\textwidth}
$$\begin{matrix}
        \star & 1 & \star & \star & 1 & \star\\
        \rightarrow &\rightarrow & \star &2 & \star & 2 & \star & \star\\
        \rightarrow & \star &3 & \star & \star & \star & 3
\end{matrix}$$
\end{minipage}

Intuitively, we desire to find a way to shift each row in a way such that every column contains at most one entry different to $\star$, 
and use this information to compress the matrix into a one-dimensional vector of row indices of length $5+2=7$. 
For instance, the shift function $s(1) = 0$, $s(2) = 2$, $s(3) = 1$ allows such a desired shift, which can be represented as shown on the right of $M$.
From this, we can obtain the placement $P=[\star{},1,3,2,1,2,3,\star{}]$ by setting $P[j]$ to be equal to the only row index which has a non-zero in column $j$ after the shift function is applied. Since such $P$ is of the desired length, the problem has a positive answer.
\end{example}

In this paper, we also focus on a new variant of the problem, which we call \probMinLength.
In this variant, we wish to minimize the length of the resulting placement $P$ by \emph{trimming the non-zero borders}, which is equivalent to performing such task on trimmed input strings. 
Our problem is formally defined as follows:

\problembox{\probMinLength{}
	\\
	\textbf{Input:}
A set of $n$ binary strings $S_1, \ldots, S_n$ such that $S_i \in i\{\star,i\}^{\ell_i} i$ for some $\ell_i \leqslant \ell - 2$ (i.e., possibly different lengths and all over distinct binary alphabets),
and an integer $\rho \in [0..\ell n]$.
\\
\textbf{Task:}
Decide whether there exists a string, called \emph{placement} $S[1..\rho] \in (\{\star\} \cup [1..n])^\rho$ such that each $S_i$ has a match at some position of $S$, where $\star$ always matches.
}
For example, the \probMinLength instance corresponding to the one given in Example~\ref{example:smc} would be $S_1 = 1\star\star 1$, $S_2 = 2\star 2$, 
$S_3 = 3\star\star\star 3$, and the resulting output  minimum-length string would be $S = 132123$.  
This problem can also be modeled as a combinatorial puzzle of 1-dimensional polyominos with gaps; as such, we also refer to the input strings as {tiles} (see \cref{section:preliminaries}). 
For more insight on the motivation for defining this variant, and other possible variations, see \cref{section:variants-relationships} in the appendix.

\subparagraph*{State-of-the-art}
We observe that the number of possible configurations for the shift function $s$ is $(\rho+1)^n$, which is prohibitive even for a maximum shift $\rho$ of~1.
Hence, there is no hope for a brute-force algorithm checking all possible shifts to be efficient.
While \probMinLength{} is a new problem,
\probMaxShift{} has been studied since the 1970s.
\citet{ziegler77small} was the first who mentioned \probMaxShift{}.
He gave a heuristic that tries to fit the next row at the leftmost possible available position.
He also augmented his heuristic with a strategy that rearranges the order of the rows by prioritizing the row with the largest number of numerals.
This strategy is known as Ziegler's algorithm.
Despite the fact that Ziegler's algorithm is only a heuristic, it is often used in practice due to its good performance.
For instance, \citet{sadayappan89efficient} similarly studied practical aspects of Ziegler's algorithm,
and \citet[Section 3.9.8]{aho86compilers} recommend using Ziegler's algorithm to represent the state transition of a deterministic finite automaton (DFA) in a compressed form.
However, the approximation ratio of Ziegler's algorithm has not been studied yet.

Unfortunately, finding the optimal solution is generally hard.
\citet{even77remarks} gave an NP-hard proof based on 3-coloring, which found an entry in the textbook of Garey and Johnson~\cite[Chapter A4.2, Problem SR13]{garey79computers}.
They showed that the problem is NP-complete even if the maximum needed shift is at most two.
\probMaxShift{} has been adapted to Bloom filters~\cite{chang91letter,chang93refinement,dillinger22fast}, 
and has also been studied under the name \textsc{Compressed Transition Matrix}~\cite[Sect.~4.4.1.3]{martin2010scientific} problem.
Finally, \citet{bannai24npcompleteness} showed that \probMaxShift{} is NP-hard even when the width $\ell$ of the matrix is $\ell \in \Omega(\lg n)$.
Their hardness proof can be applied directly to \probMinLength{}, proving that \probMinLength{} is NP-hard even when all tiles have a width of $\ell \in \Omega(\lg n)$.
That is because they reduce to a problem instance whose placement has no holes (i.e., does not contain the wildcard $\star$ and a length equal to the sum of the numerals of all tiles).
For such placements we have the following observation.

\begin{observation}\label{lemNoHolesOpt}
A placement of \probMinLength{} without holes is optimal, but not every problem instance admits a placement without holes.
\end{observation}

\citet{chang96improvement} considered a different variant in which cyclic rotations of a matrix row are allowed. However, this work is only practical and does not provide any theoretical guarantees.
Another variation is to restrict rows to have no holes and to have not only one placement, but a fixed number of placements of the same length,
which reduces to a rectangle packing problem~\cite[Section~2]{demaine07jigsaw}.
Finally, \citet{manea24subsequences} studied a variant of embedding subsequences with gap constraints.

\subparagraph*{Our Contribution} 
In this paper, we deepen the study of problem \probMaxShift, initiating a theoretical study of \probMinLength{} as a natural variation of \probMaxShift{} as well.
We first give motivation for defining \probMaxShift{} among other variation candidates.
Subsequently, we analyze the approximation ratio of previously-known greedy strategies for both problems, proving approximation ratios of $\mathcal O(\sqrt{m})$ for $\probMinLength$ and $\Omega(\sqrt{m})$ for \probMaxShift, if the optimal solution has length $m$. 
As previously mentioned, NP-hardness for \probMaxShift is known in the literature, even when the width of the matrix is $\ell \in \Omega(\lg n)$, and from this it follows that \probMinLength is also NP-hard, even when all tiles have a width of $\ell \in \Omega(\lg n)$. Still, we prove further parameterized hardness results for \probMinLength, namely, when the input is composed by only one type of tile of length $\Omega(n^2)$, and when the input is composed of tiles with exactly two non-$\star$ letters. 
On the positive side, we give a DP algorithm (\cref{thmMinLengthDPGeneral}) running in time $O(n^\kappa \ell n 2^\ell n) \subset O(n^{2^\ell} \ell n 2^\ell n)$ for \probMinLength,
where $\ell$ is the maximum length of a tile and $\kappa$ is the number of distinct tile types, which can be also adapted to solve \probMaxShift.
We present a simplified algorithm (\cref{thmminLengthPoly}) for the special case where there is only one distinct kind of tile (i.e., all tiles have the same shape). 
Finally, we experimentally compare our DP algorithm with greedy approximation algorithms using different tile-choosing strategies.

\section{Preliminaries}
\label{section:preliminaries}
A \emph{partial string} is a string that contains the \emph{wildcard symbol} $\star$ that can match any letter (including itself).
Non-wildcard symbols are called \emph{numerals}.
We say that a partial string $S$ has a \emph{match} in another partial string $P$ at position $x$ if $x+|S| < |P|$ and $S[i]= P[x+i]$ for all $i$ such that $S[i]\neq \star$ (that is, $\star$ always matches). 

\noindent
\begin{minipage}{0.5\linewidth}
	\begin{example}
$S=\star{}1\star{}11\star{}$ has a match in $P = 3321211$ at position $x=3$:
	\end{example}
\end{minipage}
\hfill
\begin{minipage}{0.38\linewidth}
\begin{equation*}
   \begin{matrix}
         i & 1 & 2 & 3 & 4 & 5 & 6 & 7\\
         S =&&& \star &1 & \star & 1 & 1 & \star\\
         P[i]=&3 &3 & 2 & 1 & 2 & 1 & 1 &3
 \end{matrix}
\end{equation*}
\end{minipage}

A \emph{tile} is a partial string that starts and ends with a numeral and contains only one distinct numeral (as in the input of problem \probMinLength).
We say that two tiles $X$ and $Y$ have the \emph{same tile type} if $X[i]$ and $Y[i]$ are both either numerals or wildcards, for all $i$. 
For instance, $1 \star 1$ and $2 \star 2$ have the same tile type, but $11$ does not.

We define the following string-formulation of \probMaxShift, which is equivalent to the latter (see \cref{section:variants-relationships} in the appendix for more details): 

\problembox{\probMaxShiftString{} {}
\\
\textbf{Input:} 
A set of $n$ binary strings $S_1, \ldots, S_n$ such that $S_i \in \{\star,i\}^{\ell}$ for some $\ell$ (i.e., all have the same length $\ell$ but  distinct binary alphabets),
and an integer $\rho \in [0..\ell n]$.
\\
\textbf{Task:}
Decide whether there exists a partial string $P[1..\ell+\rho]$ with $P[i] \in \{\star\} \cup [1..n]$ for every $i \in [1..\ell+\rho]$, such that each $S_i$ has a match in $P$ at some position $x_i \leqslant \rho$. 
}
{In the rest of the paper, we will thus equivalently refer to \probMaxShift as a problem on a matrix or on a set of strings. } 
Furthermore, we will refer to the resulting output strings of \probMinLength or \probMaxShift as a \emph{placement} of the strings/tiles.

\subparagraph*{Roadmap}
In the rest of the paper, we present results for both the original \probMaxShift{} problem, as well as our newly introduced natural variation \probMinLength{}.
For both problems, we first analyze the approximation ratio of known greedy strategies in~\cref{section:greedy}.
Then, we describe our dynamic programming algorithms to solve both problems for double-logarithmic matrix widths or logarithmic widths in~\cref{section:dp},~\cref{section:hardness} is dedicated to proving hardness results for \probMinLength{} and \probMaxShift{} with respect to different parameters. Finally, in~\cref{section:experiments} we present experimental results for the greedy algorithm, with different tile-choosing strategies, and for our dynamic programming algorithm.

\section{Approximation Ratio of Greedy Strategies}
\label{section:greedy}
As far as we are aware of, all algorithmic ideas in the literature for  \probMaxShift{} (and, consequently, \probMinLength{}) are based on the leftmost-fit greedy strategy. Indeed, Ziegler's algorithm is a special case of this type of algorithm. In this section, we study the approximation ratio of this strategy both when applied in its basic form, and in the special case of Ziegler's algorithm.

The \emph{leftmost-fit greedy strategy} processes the tiles in the order in which they are given, for instance by a queue, and proceeds as follows:
\begin{enumerate}
\item Place the first tile at position 1.
\item Place the next tile at the leftmost position $\geqslant 1$ at which the tile fits.
\item Repeat Step 2 until the last tile of the input queue has been placed.
\end{enumerate}
Ziegler's algorithm runs the leftmost-fit greedy strategy
after sorting the input sequences by their number of numerals in descending order. 

We show the following result on the approximation ratio of the leftmost-fit greedy strategy and Ziegler's algorithm, for \probMinLength and \probMaxShift:

\begin{theorem}\label{thmApproxSqrtM}
	Both greedy strategies have an approximation ratio of $\Oh{\sqrt{m}}$ for \probMinLength{} (resp.\ \probMaxShift{}) if the optimal solution has length $m$ (resp.\ maximum shift $m$).
	This ratio is tight in the sense that there is an instance for which both strategies exhibit a ratio of $\Om{\sqrt{m}}$.
\end{theorem}

	\begin{table*}[t]
		\centering
	\caption{$X$ and $Y$ for $\delta=4$ with collisions at $i_1 = 1$ and $i_2 = 13$. $x$ and $y$ are distinct values.
        }
		\label{tab:label}
			\setlength{\tabcolsep}{1pt}
		\begin{tabular}{p{3em}*{22}{p{1.5em}}}
 $i$ & 1 & 2 & 3 & 4 & 5 & 6 & 7 & 8 & 9 & 10 & 11 & 12 & 13 & 14 & 15 & 16 & 17 & 18 & 19 & 20 & 21 
	\\
 $X[i] = $  &
	$x$ & $\star{}$ & $\star{}$ & $x$ & $\star{}$ & $\star{}$ & $x$ & $\star{}$ & $\star{}$ & $x$ & $\star{}$ & $\star{}$ & $x$ & $\star{}$ & $\star{}$ & $x$ & $\star{}$ & $\star{}$ & $x$ & $\star{}$ & $\star{}$ 
	\\
	$Y[i] = $ &
	$y$ & $\star{}$ & $\star{}$ & $\star{}$ & $y$ & $\star{}$ & $\star{}$ & $\star{}$ & $y$ & $\star{}$ & $\star{}$ & $\star{}$ & $y$ & $\star{}$ & $\star{}$ & $\star{}$ & $y$ & $\star{}$ & $\star{}$ & $\star{}$ & $y$ \\
		\end{tabular}
	\end{table*}

We split the proof of \cref{thmApproxSqrtM} in two parts, for the lower and the upper bounds.

\subsection{Proof of \cref{thmApproxSqrtM}: Lower Bound}
\label{section:approx-lowerbound}
Let us consider an instance of \probMinLength where there are two different types of tiles, $X$ and $Y$.
Each tile of type $X$ (resp. $Y$) is a periodic string with root $x \mapsto x{\star}^{\delta-2}$  (resp.\ $y \mapsto y {\star}^{\delta-1}$), where $x$ (resp. $y$) is the respective numeral.
Let the input be formed by $\delta-2$ tiles of type $X$ and $\delta-1$ of type $Y$ (see Table~\ref{table:example-lower-bound} in the appendix for an example with $\delta=4$).

Since $\delta-1$ and $\delta$ are co-prime, 
$X[i]$ and $Y[i]$ are numerals if and only if $i = 1 + jk(\delta-1)$ for some integer $j \geqslant 0$ and $j \leqslant \min(|X|,|Y|)$.
Let us suppose that $X$ and $Y$ have lengths of at least $1+\delta(\delta-1)$.
We can then say more generally that, for every $d \in [0..\delta]$, there exists an $i \in [1..\delta(\delta-1)]$ such that $X[i]$ and $Y[i+d]$ are numerals.

Thus, if we placed a tile of type $Y$ (resp. $X$) at position $1$, the first next position at which we can place a tile of type $X$ (resp. $Y$) is in the range $[|Y|-2k..|Y|]$ (resp. $[|X|-2k..|X|]$). 
Suppose that we have placed a tile $Y_1$ of type $Y$ first, and a tile of type $X$ next at the first available position.
Furthermore, suppose that we now want to place another tile $Y_2$ of type $Y$. 
Since two tiles of the same type do not match, the greedy algorithm shifts $Y_2$.
However, we need to shift it to the last $\delta$ positions of the placed tile of type $X$, otherwise there is a mismatch since some positions of $Y_2$ are in conflict with the already placed $X$.
If we continue placing tiles of type $X$ and $Y$ in alternating order, we end up with a placement of length $\Om{ (\delta-2)(|X| - 2k) + (\delta-1)(|Y| - 2k) }$.
Setting the lengths of $X$ and $Y$ to $\delta(\delta+1)+1$, the placement length is $\Om{\delta^3}$.

However, an optimal placement has length $\Ot{\delta^2}$. 
To see this, first place all $\delta-2$ tiles of type $X$ at positions $1$, $2$, $\ldots$, $1+\delta-2$.
This gives a placement of length $1+\delta-2 + |X|-1 = 1+\delta-2 + \delta(\delta+1)+1-1 = \delta^2+2k-1$.
Next, place all tiles of type $Y$ at subsequent positions $\delta^2+2k$, $\delta^2+2k+1$, $\ldots$, $\delta^2+3k-2$.
The total length of this placement thus is $\delta^2+3k-2 + |Y| = \Ot{\delta^2}$.
Since this placement is without holes, it is optimal by \cref{lemNoHolesOpt}.
In total, the placement of the greedy algorithm is at least $\Om{\delta}$ times larger than the optimal solution of length $\Ot{\delta^2}$.
Since $\delta$ is arbitrary, the greedy algorithm selecting tiles of alternating type $X$ and $Y$ has an approximation ratio of at least $\sqrt{m}$, where $m = \Ot{\delta^2}$ being the length of the optimal solution.

Finally, we can scale the lengths of each $X$ and $Y$ tile individually to let Ziegler's method behave like the leftmost-fit greedy strategy.
Since Ziegler's method greedily selects the remaining tile with the largest number of numerals,
we adjust the length of every $X$ and $Y$ tile individually such that the method will select $X$ and $Y$ in an alternating order.
To this end, define $X_x = (x{\star}^{\delta-2})^{2x+1} x$ and 
$Y_y = (y {\star}^{\delta-1})^{2y} y$.
This lets Ziegler's precomputation step arrange the tiles in a list $[Y_{\delta-1}, X_{\delta-2}, Y_{\delta-2}, X_{\delta-3}, \ldots, X_1, Y_1]$.
Like before, an optimal solution is to combine all $X_x$ and $Y_y$, each separately into two placements without holes, each of length $\Ot{\delta^2}$.
Ziegler's method here works like the alternating algorithm, resulting in a string of length $\Om{\delta^3}$. 

\subparagraph*{Approximation Ratio for \probMaxShift}
From the previous discussion, we can derive the same approximation ratio for \probMaxShift as well:
If we write each tile as a row in a matrix (appending $\bsq{\star}$s at the end of shorter tiles to make them all of the same length), then both greedy strategies have a maximal shift of $\Om{\delta^3}$. 
The solution without holes needs a maximal shift of $\Oh{\delta^2}$ to move the last $Y$.
Therefore, the approximation ratio is at least $\Om{\sqrt{m}}$.

\subsection{Proof of \cref{thmApproxSqrtM}: Upper Bound}
\label{section:approx-upperbound}
We start with \probMinLength{}, and prove the upper bound by contradiction. 
Given that the length of the optimal solution is $m$,
assume by contradiction that one of the greedy strategies generates a placement $P$ longer than $(3m+1)\sqrt{m} \in \Ot{m \sqrt{m}}$. 
 The coarse idea is to split $P$ into two strings $U$ and $V$ such that we can show that $U$ or $V$ have at least $m$ many numerals, which contradicts the optimal length $m$ by the pigeonhole principle.
 In detail, we split $P$ into a prefix $U$ of $2m\sqrt{m}$ positions and a suffix $V$ of 
 $(m+1)\sqrt{m}$ positions such that $P = U \cdot V$.
 We further decompose $U$ into chunks, each of length $2m$.
 Let $C$ denote the chunk in $U$ with the minimum number of numerals, which we define to be $u$.

We consider tiles that the greedy strategy could not place in $U$ (nor at the border of $U$ and $V$) and thus
got placed completely into $V$. 
 Let $P$ be one of these tiles with the least number of numerals, which we define as $v$.
 Since the optimal solution has length $m$, all tiles have lengths of at most $m$.
 In particular, $P$ is completely contained in $V$, which is composed of at least $\sqrt{m}+1$ tiles.

Since $P$ could not be placed in $U$, $P$ cannot match at any starting position in $C$.
Each mismatch can be expressed by the ranks of numerals of the $u$ numerals in $C$ and the $v$ numerals in $P$.
Using a counting argument for the ranks $\in [1..u] \times [1..v]$, the product $xy$ must be at least as large as $2m$, the length of $C$.
In particular, $u$ or $v$ must be at least $\sqrt{m}$, for which we have a short case analysis.

 \begin{itemize}
 	\item If $u$ is at least $\sqrt{m}$, then $U$ must have at least $m$ numerals.
		That is because it has $\sqrt{m}$ chunks and each chunk has at least as many numerals as $C$ by choice of $C$.
	\item If $v$ is at least $\sqrt{m}$, then $V$ must have at least $m$ numerals.
		That is because at least $\sqrt{m}$ tiles have been placed into $V$, 
		and each has at least as many numerals as $P$ by choice of $P$.
 \end{itemize}

In both cases, the number of numerals in $P$ must exceed $m$.
We thus obtain a contradiction that $P$ and the optimal solution must contain the same number of numerals, so the optimal solution must be longer than $m$.

\subparagraph*{Approximation Ratio for \probMaxShift}
We obtain the approximation ratio $O(\sqrt{m})$ for \probMaxShift{} by following the same steps as above by letting $m$ denote the maximum shift of the optimal solution and consider the shift of $P$.

\section{Dynamic-Programming Algorithms}
\label{section:dp}

We present our dynamic programming algorithms for \probMinLength (which can also be adapted to \probMaxShift). 
The main ingredient for our dynamic programming algorithms is a pattern matching algorithm that respects wildcards,
which we use in the following lemma for finding all positions at which we can merge a tile with a placement.

\begin{lemma}\label{lemPartialWordMatching}[Pattern Matching on Partial Words~\cite{clifford07simple}]
    Given a placement $P$ of length $n$ and a tile $T$ of length $m$, 
    we can find in $\Oh{n \lg m}$ time all positions of $P$ at which we can insert $T$.
\end{lemma}
A crucial observation is that we find these insertion positions even if we replace all numerals in the placement $P$ by a new numeral $-1$ that does not match with any other numeral (but matches $\star$).
Thus, in the following, it suffices to regard $P$ as a binary string of the alphabet $\{-1, \star\}$.

We first present a DP algorithm for the special case where we only have one type of tile, showing that in this case the problem becomes polynomial-time solvable when the tile length is logarithmic in the number of tiles. Then, in \cref{section:dp-general} we present the algorithm for the general case. 

\subsection{One Distinct Tile with Logarithmic Length}\label{secDPLogLength}

In this section, we assume that all tiles have the same tile type.
With that restriction, we can devise a DP-algorithm that solves \probMinLength{} in time polynomial in the number of tiles $n$, if the tile length $\ell$ is logarithmic in $n$.
Let us consider all different placements we can produce with $i$ tiles for increasing $i \in [1..n]$:
each placement has a length of at most $i\ell$, and their number  can be exponential ($\leqslant 2^{i\ell}$), which is prohibitive.
The insight is that we only need to keep track of the last $\ell$ bits of each placement for a fixed $i$, if we prolong a placement always to its right. 
Note that this choice is without loss of generality (by symmetry).
So, given a fixed placement $P$, we try to obtain all possible different $\ell$-length suffixes by combining $P$ with a newly chosen tile.
Among all the placements with $i$ tiles having the same $\ell$-length suffix, we only keep the shortest one.
Thus, for each $i$, we produce $O(2^{\ell})$ placements from the placements of $i-1$ (or the empty string if $i=1$).
For $i = n$, we can report the shortest placement among all produced ones.

\subparagraph*{Algorithmic Details.}
In concrete terms,
the DP table is a table $X[0..n][0..2^\ell-1]$.
For an integer $i \in [1..n]$ and an $\ell$-length bit vector $Y$,
$X$ obeys the invariant that $X[i][Y] = \lambda > 0$ 
if and only if there is a placement of length $\lambda$ using $i$ tiles having an $\ell$-length suffix equal to $Y$, where $\lambda$ is chosen to be minimal.
We start with $X[0][0] = 0$, and fill $X[i][\cdot]$ as follows:
For each defined value $X[i-1][Y]$, compute all possible $\ell$-length suffixes $Z = Y \bowtie T$, where $T$ is the input tile.
The set $X \bowtie Y$ contains all placements $P$ of $X$ and $Y$ such that $X$ has an earlier match than $Y$ in $P$ and there is no gap between the matching substrings of $X$ and $Y$ in $P$.
Formally, for two strings $A$ and $B$, let $\ell \in [0..|A|]$ such that
for all $i \in [\ell+1..|A|]$, $A[i] = \star{}$ or $B[i-\ell] = \star$ holds.
Written the other way around, there is no $i \in [\ell+1..|A|]$ such that both $A[i]$ and $B[i-\ell]$ are numerals.
For each such $\ell$ define
$S_\ell[1..\max(|A|,|B|+\ell)]$ such that
\[
S_\ell[i] = 
\begin{cases}
    A[i] &\text{~if~} A[i] \neq \star \wedge i \in [1..|A|], \\
    B[i-\ell] &\text{~if~} B[i-\ell] \neq \star \wedge i \in [\ell+1..\ell+|B|], \\
    0 &\text{~else},
\end{cases}
\]
for every $i \in [1..\max(|A|,|B|+\ell)]$.
Then $A \bowtie B := \{ S_\ell \}_\ell$ is the set of all such constructed $S_\ell$.
By \cref{lemPartialWordMatching}, we can compute $A \bowtie B$ in $\Oh{|A| \lg |B|}$ time.

\begin{example}
	For $A = \texttt{aa$\star{}$a$\star{}$a}$ and $B = \texttt{b$\star{}$b}$,
	we have $A \bowtie B = \{ \texttt{aababa}, \texttt{aa$\star{}$abab}, \texttt{aa$\star{}$a$\star{}$ab$\star{}$b} \}$.
\end{example}

Let $d$ be the number of the new positions added to the placement by merging $Y$ and $T$:
if $X[i][Z]$ is undefined or $X[i][Y] + d < X[i-1][Z]$, we set $X[i][Z] \gets X[i-1][Y] + d$.
After filling $X[i][\cdot]$ we recurse for increasing $i$ until $i = n$.
Finally, in $X[n][\cdot]$, we return the minimum defined value in $X[n][\cdot]$ as the length of the optimal placement using all $n$ tiles.
The steps are summarized in pseudocode in \cref{algminLengthDP} in the appendix.

\subparagraph*{Correctness}
The correctness of the algorithm follows by induction on $i$.
For $i=0$, the only possible placement is the empty placement of length $0$ with an $\ell$-length suffix of $\bsq{\star}$s.
We here interpret the suffix as the non-trimmed placement prior to removing all surrounding positions with $\bsq{\star}$s. 
Suppose that the invariant holds for $i-1$.
Then, for each defined $X[i-1][Y]$, we compute all possible $\ell$-length suffixes $Z$ by merging $Y$ with the tile $P$. 
Now suppose that there is actually a placement $P$ of $i$ tiles with a suffix of $Y$ using less than $X[i][Y]$ positions.
Since it is possible to extract the placement of the rightmost tile $T$ in $P$, we can remove $T$ from $P$ and obtain
a placement of $i-1$ tiles with a suffix of some $Y'$ such that merging $Y'$ with $T$ gives $Y$.
However, by the induction hypothesis, we have $X[i-1][Y'] \leqslant |P| - d$, where $d$ is the number of new positions added by merging $Y'$ and $T$, a contradiction.
Informally, we cannot omit a possible solution by only tracking the $\ell$-length suffixes of placements since
otherwise we would have already considered a prefix of that placement when processing less tiles.

\subparagraph*{Complexity}
We compute $Z$ from $Y$ and the tile $P$ by leveraging \cref{lemPartialWordMatching} to find all positions at which we can merge the tile with the placement represented by $Y$.
Since $Y$ has a length of at most $\ell$, we need $O(2^{\ell} \cdot \ell \log \ell)$ time for each $i$, thus $O(n \ell 2^\ell \log \ell)$ total time.
As a side note, it suffices to keep only the latest computed row $X[i][\cdot]$ in memory, so the space is $O(2^{\ell})$.
We conclude that \probMinLength{} is in polynomial time solvable if all tiles are non-distinct and $\ell = O(\lg n)$:

\begin{theorem}\label{thmminLengthPoly}
	\probMinLength{} with only one distinct tile of length $\ell$ is solvable in $O(n^2 2^\ell \ell \log \ell)$ time.
	In particular, it is fix-parameter tractable in $\ell$ and polynomial-time solvable if $\ell = O(\lg n)$.
\end{theorem}

\subsection{General DP-algorithm}
\label{section:dp-general}
Here we present the DP algorithm for a general instance of \probMinLength{}.
The main obstacle in generalizing \cref{thmminLengthPoly} to more tile types is that the order in which we choose the tiles now becomes important.

Let us assume that there are $\kappa$ tile types in our instance, 
each with cardinality $c_i \geqslant 1$ for $i \in [1..\kappa]$ (i.e., there are $c_1$ tiles of the first kind, $c_2$ tiles of the second, and so on).
Then, the number of ways to sort all $n$ tiles is $\frac{n!}{c_1! c_2! \cdots c_\kappa!}$, which is exponential in $n$ for general $\kappa$.
However, as we will see, it suffices to only track 
the \emph{Parikh vector} of the used tiles instead of their order. 
A Parikh vector of a multiset of tiles is a vector of length $\kappa$ where the $i$-th entry counts the number of tiles of type $i$ in the multiset.
Let $\vec{p}_0 = (0,0,\ldots,0)$ be the empty Parikh vector of length $\kappa$ and $\vec{p}_n = (c_1, c_2, \ldots, c_\kappa)$ be the Parikh vector of all input tiles.

As before, we prolong a partial solution to the right by adding a new tile of a given tile type at the $\ell$-length suffix.
We then track partial solutions for each Parikh vector individually:
instead of considering a table of the form $X[i][\cdot]$ for integer $i$, we substitute $i$ with a Parikh vector $\vec{p}$ representing the amounts of types of tiles added so far, and 
we fill $X[\vec{p}][\cdot]$ using all solutions $X[\vec{p'}][\cdot]$ such that $\vec{p}$ can be obtained by adding one tile of a given tile type to $\vec{p'}$.
In particular, $|\vec{p'}| = |\vec{p}| - 1$.
By doing so, $X$ holds the following invariant, which can be proven by induction on $|\vec{p}|$:
given $X[\vec{p}][Y] = \lambda > 0$, there is a placement of length $\lambda$ using the multi-subset of tiles represented by $\vec{p}$ having an $\ell$-length suffix equal to $Y$, where $\lambda$ is chosen to be minimal.
This invariant leads to the correctness of the algorithm, 
which outputs the minimum defined value in $X[\vec{p}_n][\cdot]$ as the length of the optimal placement using all $n$ tiles.
We summarize the steps in pseudocode in \cref{algminLengthDPGeneral}.

\subparagraph*{Correctness}
The correctness follows by adapting the induction proof for \cref{thmminLengthPoly}.
Instead of performing induction on $i$, we do so on $|\vec{p}|$.
Given that the induction invariant holds for all $\vec{p'}$ with $|\vec{p'}| \leqslant |\vec{p}|-1$,
suppose that there is actually a placement $P$ of tiles drawn as specified by $\vec{p}$ with a suffix of $Y$ using less than $X[\vec{p}][Y]$ positions.
Since it is possible to extract the placement of the rightmost tile $T$ in $P$, we can remove $T$ from $P$ and obtain 
a placement of the multi-subset of tiles represented by some $\vec{p'}$ with a suffix of some $Y'$ such that merging $Y'$ with $T$ gives $Y$.
However, by the induction hypothesis, we have $X[\vec{p'}][Y'] \leqslant |P| - d$, where $d$ is the number of new positions added by merging $Y'$ and $T$, a contradiction.
Consequently, selecting among the same multiset of tiles the shortest placement yields optimality, independently of their order. 

\subparagraph*{Complexity}
There are $n$ tiles, and each tile can be represented by a binary number with $2^\ell$ bits.
Thus, the number of different tile types $\kappa$ is at most $2^\ell$.
A Parikh vector thus has length at most $2^\ell$.
Since a number in a Parikh vector has the domain $[1..n]$, 
the count of all Parikh vectors is upper bounded by $n^{2^\ell}$.
For each Parikh vector $\vec{p}$, we compute $O(2^{\ell})$ suffixes $Z$ from each defined $X[\vec{p'}][Y]$ using \cref{lemPartialWordMatching}.
Thus, the total time is 
$O(n^\kappa \ell n 2^\ell n)$.
Like for one tile-type case, for the computation of all rows $X[\vec{p}][\cdot]$ with $|\vec{p}| = i$, it suffices to keep only the rows $X[\vec{p'}][\cdot]$ with $|\vec{p'}| = i-1$ in memory.
However, there can be many such Parikh vectors with the same length.
While the trivial maximum number is $O(n^\kappa)$, the exact number to compute is nontrivial, cf.~\cite{gluck20integer}.
We conclude that \probMinLength{} is in polynomial time solvable if all tiles are non-distinct and $\ell = O(\lg n)$:

\begin{theorem}\label{thmMinLengthDPGeneral}
\probMinLength{} with $n$ tiles of length $\ell$ drawn from $\kappa$ tile types is solvable in $O(n^\kappa \ell n 2^\ell n) \subset O(n^{2^\ell} \ell n 2^\ell n)$ time.\
It is polynomial-time solvable if $\ell = O(\lg \lg n)$, or if
$\ell = O(\lg n)$ and $\kappa = O(\lg n)$.
\end{theorem}

\begin{algorithm}[t]
\caption{DP-algorithm for \probMinLength{} with distinct tiles.}
\label{algminLengthDPGeneral}
\begin{algorithmic}[1]
\State \textbf{Function} minLengthDPGeneral($n$ tiles of length at most $\ell$)
\State Compute $\vec{p}_n$ from the input tiles, let $T_1, T_2, \ldots, T_\kappa$ be the different tile types
\State initialize $X[\vec{p}_0][0] \gets 0$
\For{each Parikh vector $\vec{p}$ with $X[\vec{p}][Y]$ defined and $|\vec{p}|$ from $0$ to $n-1$}
	\For{each $i \in [1..\kappa]$ with $\vec{p}[i] < \vec{p_n}[i]$}
		\State let $\vec{p'}$ be $\vec{p}$ with $\vec{p'}[i] = \vec{p}[i] + 1$
		\For{each distinct $\ell$-length suffix $Z$ among the strings in $Y \bowtie P$}
			\State Let $d$ be the number of new positions added by merging $Y$ and $P$
			\If{$X[\vec{p'}][Z]$ is undefined or $X[\vec{p}][Y] + d < X[\vec{p'}][Z]$}
				\State $X[\vec{p'}][Z] \gets X[\vec{p}][Y] + d$
			\EndIf
		\EndFor
	\EndFor
\EndFor
\State \Return the minimum defined value in $X[\vec{p}_n][\cdot]$
\end{algorithmic}
\end{algorithm}

\subparagraph*{Implementation Details}
Instead of creating the whole DP-matrix $X$ as a two-dimensional array,
we can implement $X$ by hash tables $H_i$ representing $X[\vec{p}][\cdot]$ for all $\vec{p}$ with $|\vec{p}| = i$.
In fact, we only need to keep two hash tables $H_i$ and $H_{i+1}$ in memory at the same time.
Computing $H_{i+1}$ from $H_i$ can be done as with $X$ by scanning all entries in $H_i$.

\subparagraph*{Adaptation to \probMaxShift{}}
The DP algorithm for \probMinLength{} can be adapted to solve \probMaxShift{} as well.
The change is straightforward: instead of storing the minimum length of a placement in $X[\cdot][\cdot]$, we store the minimum maximal shift of a placement.
When merging two placements, we compute the maximal shift of the merged placement and update $X[\cdot][\cdot]$ accordingly.

\section{Parameterized Hardness}
\label{section:hardness}

It is natural to study parameters for which \probMinLength{}, and \probMaxShift{}, can become easier to compute.
This section is dedicated to providing a full study of the parameterized complexities of these problems with respect to several natural parameters, namely: 
\begin{itemize}
	\item the maximum length $\ell$ of a tile, 
	\item the maximum number $\zeta$ of numerals in a tile,
	\item the number of tile types $\kappa$, \item the length $m$ of the optimal placement  (for \probMinLength), and
	\item the maximum shift $\rho$ in an optimal solution (for \probMaxShift).
\end{itemize}
\cref{table:parameterized-results} shows a summary of all parameterized results for the two problems.

\begin{table}[t]
	\begin{minipage}{0.42\linewidth}
\caption{Summary of positive and negative parameterized results for \probMinLength and \probMaxShift. 
The parameters denote the following. $\ell$ is the maximum length of a tile, $\zeta$ is the maximum number of numerals per tile, $\kappa$ is the number of tile types, $m$ is the optimal placement size, and $\rho$ is the maximum shift.  $\dag$ or $\mathparagraph$ denotes results only concerning \probMinLength{} or \probMaxShift{}, respectively.}
\label{table:parameterized-results}
	\end{minipage}
	\hfill
	\begin{minipage}{0.55\linewidth}
\begin{tabular}{l || l | l | }& {Polynomial-time for} & {NP-hardness even for} 
\\ \hline\hline 
	\multirow{2}{*}{$\ell$}   & $\ell \in O(\lg \lg n)$ or  (\cref{thmMinLengthDPGeneral}) & \multirow{2}{*}{$\ell \in \Omega(\lg n)$~\cite{bannai24npcompleteness}}
\\ 
				  & $\ell \in O(\lg n) \wedge \kappa \in O(\lg n)$ &
\\ \hline

	\multirow{2}{*}{$\zeta$}      & \multirow{2}{*}{$\zeta=1$ (trivial)}                                                        & {$\zeta=2$ (\cref{thmNP-hard-two-numerals}) $^\dag$} 
	\\
				      & & {$\zeta=4$ (\cref{thmNP-hard-three-numerals}) $^\mathparagraph$}
\\ \hline

$\kappa$ & $ \kappa \in O(\lg n) \wedge \ell \in O(\lg n)$                           & {$\kappa=1$ (\cref{thmNP-hardOneTile})} \\ \hline

$m$      
         & $m = \ell$ (trivial) $^\dag$                                                  & $m\leqslant \ell +  2$~\cite{even77remarks}
\\ \hline

$\rho $     & $\rho=0$ (trivial) $^\mathparagraph$                                                       & $\rho\leqslant 2$~\cite{even77remarks} $^\mathparagraph$  \\

\end{tabular}
	\end{minipage}
\end{table}

Hardness with respect to some of these parameters follows from previous results. It is known that \probMinLength{} and \probMaxShift{} are NP-hard even when the width $\ell$ of the matrix is $\ell \in \Omega(\lg n)$~\cite{bannai24npcompleteness}.
\probMaxShift{} is also NP-hard when the maximum shift $\rho$ is at most two~\cite{even77remarks}, which translates to hardness even if $m \leqslant \ell + 2$ for \probMinLength{}.

The remaining parameters to study are then the maximum number $\zeta$ of numerals in a tile, and the number of tile types $\kappa$.
Unfortunately, in what follows, we show hardness results for both parameters for \probMinLength{}. 

\subsection{One distinct tile (parameter $\mathbf{\kappa = 1}$)}
\label{section:NP-hard-one-tile}

We show that \probMaxShift{} is NP-hard even if there is only one type of tile, of length $\Omega(n^2)$. 
This implicitly shows the same hardness for \probMinLength{} since the two problems \probMinLength and \probMaxShift are equivalent 
when there is only one tile type.

More formally, we prove the following:
\begin{theorem}\label{thmNP-hardOneTile}
	\probMinLength{} and \probMaxShift are NP-hard even if the input is composed of one tile type of length $\ell = \Omega(n^6)$. 
\end{theorem}

Let the tile type be $T[1..\ell]$. 
In what follows we show that combining $n$ tile instances of $T$ (with different numerals) into a placement such that each of them starts within the first $\ell$ positions in the placement is NP-hard. 
We reformulate our problem as follows.

\newcommand{\probMinDisjointShifts}{\textsc{minDisjointShifts}\xspace}
\problembox{\probMinDisjointShifts{}
	\\
	\textbf{Input:} Three positive integers $n$, $\rho$, and $\ell$ with $n$ such that $n \leqslant \rho \leqslant \ell \leqslant (2n)^6$ and 
	a tile type $T[1..\ell]$ representing the set $X = \{i \in [0..\ell-1] \colon T[i+1] \neq \star \}$.
	\\
	\textbf{Task:} Decide whether there exist non-negative integers $a_1, a_2, \ldots, a_n \leqslant \rho$ such that the $n$ sets $a_i + X := \{ a_i+x \ | \ x \in X  \}$ are pairwise disjoint.
}

We prove below that \probMinDisjointShifts{} is NP-complete.
This implies the hardness of \probMaxShift as well: 
an instance of \probMinDisjointShifts{} as defined above has a positive answer if and only if $n$ tiles of the same tile type $T$ and maximal shift parameter $\rho$
are an instance of \probMaxShift{} with a positive answer. 
Indeed, $X$ represents the positions of $T$ that are numerals; 
thus, $a_i+X$ corresponds to shifting all positions of $X$ by the same amount (or, equivalently, shifting an occurrence of $T$ by $a_i$). 
Furthermore, requiring  $(a_i + X) \cap (a_j+X) = \emptyset$ is equivalent to requiring no collisions between the numerals of different tiles. 
Therefore, the integers $a_i \le \rho$ can be seen as shifts of each occurrence of the tile such that there are no numerals' collisions, i.e., the task of \probMaxShift{}.

First, \probMinDisjointShifts{} is in NP:
If such integers $a_i$ exist, it suffices to guess them, and then store in a bit vector of length $\ell+\rho$ those elements that belong to any set $a_i + V$, and to check that they are pairwise distinct.

Conversely, our NP-hardness proof is a variant of a result from~\cite{codenotti98hardness} and consists in reducing the clique decision problem to our own decision problem.

\newcommand{\probClique}{\textsc{Clique}\xspace}
\problembox{\probClique{}
	\\
	\textbf{Input:} An undirected graph $\mathcal{G} = (V,E)$ with vertex set $V = [0..v-1]$ and a parameter $n \in [1..v]$.
	\\
	\textbf{Task:} Decide whether $\mathcal{G}$ has a clique of size $n$.
}

Let $\mathcal{G} = (V,E)$ be an instance of \probClique{} with $V = [0..v-1]$ and $n \in [1..v]$.
For $\rho = 2^{3 \lceil \log_2(v) \rceil}$, \citet[Lemma~1]{codenotti98hardness} showed that there exist integers $b_0 < b_1 < \ldots < b_{v-1}$ in the set $[0..\rho-1]$ that can be computed in time polynomial in $v$, and for which the sums $b_i + b_j$ and $b_i + b_j + b_k$ are pairwise distinct.

Now, consider the set $F = \{b_i - b_j \colon (i,j) \in E\}$; since the sums $b_u + b_v$ are pairwise distinct, so are the differences $b_u - b_v$. 
Hence $F$ has cardinality $|E|$.
Let $\overline{f}_0 < \overline{f}_1 < \cdots < \overline{f}_{\mu-1}$ be the elements of the set
$\overline{F} := [1..\rho] \setminus F$ of cardinality $\mu := \rho - |E|$.
Finally, we define the set 
\(
X := \bigcup_{u \in [0..\mu-1]} \{2 u \rho \} \cup \{2 u \rho + \overline{f}_u \} \subset [0..2 \rho \mu-1],
\)
and observe that we can represent $X$ with a tile $T$ of length $\ell = 2 \rho \mu$ and with numerals at positions in $X$.
It remains to prove that the triple $(n,\rho,\ell)$ for $\ell := 2 \rho \mu$ and the set $X$ form a positive instance of \probMinDisjointShifts{} if and only if $\mathcal{G}$ has a clique of size $n$.
For that, we observe that, by construction, the set $X-X$ of pairwise differences of elements in $X$ consists in the integer $0$, the elements of $\overline{F}$, and integers larger than $\rho$.

Suppose that $\mathcal{G}$ has a clique $\{v_1,v_2,\ldots,v_n\}$ of size $n$ with $v_1 < v_2 < \cdots < v_n$. 
Let $a_i = b_{v_n} - b_{v_i}$ for all $i \in [1..n]$.
The integers $a_i$ are smaller than $\rho$ and the sets $a_i + X$ are pairwise disjoint.
Indeed, assume for the sake of contradiction that there exist integers $x,y \in X$ and indices $i > j$ such that $a_i + x = a_j + y$, i.e., $b_{v_i} - b_{v_j} = a_j - a_i = x - y$.
Hence, $b_{v_i} - b_{v_j}$ is smaller than $\rho$ and belongs to $X-X$, which means it belongs to $\overline{F}$, which is absurd since $(v_i,v_j) \in E$.

Conversely, assume that there exist integers $c_1 > c_2 > \cdots > c_n \geqslant 0$, smaller than $\rho$, for which the sets $c_i + X$ are pairwise disjoint;
up to subtracting $c_n$ from each of them, we further assume that $c_n = 0$. 
By construction, each difference $c_i - c_j$ is smaller than $\rho$ but does not belong to $X-X$, i.e., belongs to $F$.
In particular, setting $j = n$ proves that $c_i$ belongs to $F$, i.e., is a difference $b_{\mathsf{s}(i)} - b_{\mathsf{t}(i)}$
for some mappings $\mathsf{s}, \mathsf{t} : [1..n] \to [0..v-1]$.

In what follows we show that one of the mappings $\mathsf{s}$ or $\mathsf{t}$ is constant and the other maps to a clique of size $n$ in $\mathcal{G}$.
Given any two indices $i < j$ smaller than $n$, the difference $c_i-c_j$ also belongs to $F$, i.e., $c_i - c_j = b_{\mathsf{s}(i)} - b_{\mathsf{t}(i)} - b_{\mathsf{s}(j)} + b_{\mathsf{t}(j)} = b_y - b_z$ for some indices $y, z$.
This equality rewrites as $b_{\mathsf{s}(i)} + b_{\mathsf{t}(j)} + b_z = b_{\mathsf{t}(i)} + b_{\mathsf{s}(j)} + b_y$, and therefore the multisets $\{\mathsf{s}(i),\mathsf{t}(j),z\}$ and $\{\mathsf{t}(i),\mathsf{s}(j),y\}$ coincide with each other.
Recalling that $c_i \neq 0$, $c_j \neq 0$ and $c_i \neq c_j$ shows that $\mathsf{s}(i) \neq \mathsf{t}(i)$, $\mathsf{s}(j) \neq \mathsf{t}(j)$ and $b_y \neq b_z$.
It follows that either $\mathsf{s}(i) = \mathsf{s}(j)$ and $(y,z) = (\mathsf{t}(i),\mathsf{t}(j))$ or $\mathsf{t}(i) = \mathsf{t}(j)$ and $(y,z) = (\mathsf{s}(i),\mathsf{s}(j))$.
In particular, repeating this argument for all pairs $(i,j)$ proves that either all indices $\mathsf{s}(i)$ coincide and each pair $(\mathsf{t}(i),\mathsf{t}(j))$ belongs to $E$, 
or all indices $\mathsf{t}(i)$ coincide and each pair $(\mathsf{s}(i),\mathsf{s}(j))$ belongs to $E$.
In both cases, $\mathcal{G}$ has a clique of size $n$.

\subsection{Constant number of numerals per tile (parameter $\mathbf{\zeta}$)}
\label{section:NP-hardness-two-numerals}

We reduce from the scheduling problem of coupled tasks~\cite{chen21scheduling} to \probMinLength{}.

\newcommand{\probCoupledTask}{\textsc{coupledTaskScheduling}}
\problembox{\probCoupledTask{}
	\\
	\textbf{Input:} $n$ tasks, each task $i$ consisting of two sub-tasks of length $a_i$ and $b_i$ separated by an exact delay of $\ell_i$ and a makespan $C_{\text{max}}$.
	\\
	\textbf{Task:} Decide whether there exists a schedule of the $n$ tasks on a single machine such that no tasks overlap and the makespan is at most $C_{\text{max}}$.
}

Given an instance to \probCoupledTask{}, we create an instance of \probMinLength{} 
by creating a tile $S_i = y^{a_i} \star^{\ell_i} y^{b_i}$ of length $a_i + b_i + \ell_i$ for each task $i$.
Then this instance has a solution for parameter $C_{\text{max}}$ if and only if the \probCoupledTask{} instance has a solution with makespan at most $C_{\text{max}}$.
\citeauthor{chen21scheduling} showed that the problem is NP-hard even if all $a_i$ and $b_i$ are equal to 1, but the $\ell_i$'s are arbitrary (non-constant).
Therefore, we obtain our claim for \probMinLength{}.

\begin{theorem}\label{thmNP-hard-two-numerals}
	\probMinLength{} is NP-hard even if $\zeta = 2$, i.e., 
	the $y$-th tile is of the form $y \star^{\ell_y} y$ for some $\ell_y \in [1..\ell-2]$ and every $y \in [1..n]$. 
\end{theorem}

As a side note, the case of a single tile type $y \mapsto y^a \star^{\ell} y^b$ for constant $a,b,\ell$ is polynomial-time solvable~\cite[Theorem~8]{chen21scheduling}.
However, the problem remains open for the following cases:
\begin{itemize}
	\item two tile types of the form $y \mapsto y \star^{\ell_1} y$ and $y \mapsto y \star^{\ell_2} y$ for two arbitrary $\ell_1,\ell_2$, cf.~\cite{bekesi22first}
	\item one tile type of the form $y \mapsto y^a \star^\ell y^b$ for $a,b,\ell \in \omega(1)$.
\end{itemize}

For \probMaxShift{}, we can follow the original NP-hardness proof of \citet{even77remarks} (see~\cite{bannai24npcompleteness} for a description) from 3-coloring, 
where a reduction instance from a graph $(V,E)$ with $|V| = n$ represents a vertex by a tile that has as many numerals as its degree.
By \citet{garey76some}, 3-coloring is NP-hard even if each vertex has degree at most 4, and thus we obtain NP-hardness even if the tile types have at most 4 numerals.

\begin{theorem}\label{thmNP-hard-three-numerals}
	\probMaxShift{} is NP-hard even if $\zeta = 4$.
\end{theorem}

\begin{figure}[t]
	\centering
	\includegraphics[page=1]{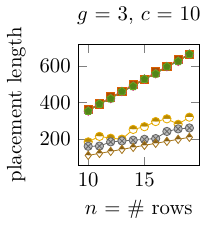}
	\includegraphics[page=1]{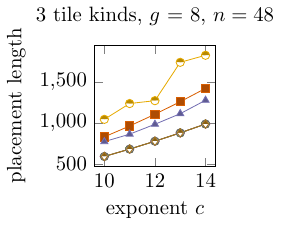}
	\includegraphics[page=2]{./plot/sameones.pdf}
	\caption{Placement lengths for different heuristics studied in \cref{section:experiments}.
		Left: Input of the form $X,Y,X,Y,\ldots$ with two tile types.
		Right: Input of the form $X,Y,Z,X,Y,Z,\ldots$ with three tile types, and additional helper tiles to fill all gaps in the shortest solution.
		All tiles $X$, $Y$ and $Z$ have the same number of ones.
	}
	\label{fig:sameones}
\end{figure}

\section{Experiments}
\label{section:experiments}

Finally, we present some experimental results concerning the output placement length of different \probMinLength heuristics, compared to the exact DP algorithm as well. 

The point of freedom in Ziegler's algorithm lies in the order in which tiles are placed.
While Ziegler suggests sorting tiles by decreasing frequencies of numerals,
we here experiment with different sorting strategies to see their effect on the final placement length.
We tested the following sorting strategies:
\texttt{none}: no sorting (the classic leftmost-fit greedy strategy),
	  \texttt{incFreq}: sort by increasing number of numerals,
	  \texttt{decFreq}: sort by decreasing number of numerals,
	  \texttt{incDens}: sort by increasing density, given by $\frac{\text{number of numerals}}{\text{length}}$, and
	  \texttt{decDens}: sort by decreasing density.
Additionally, we implemented a strategy, named \texttt{random}, that calls Ziegler's algorithm 10 times with random shufflings of the input tiles and returns the shortest result.

Unfortunately, except \texttt{random}, all strategies performed badly on the following instance,
using only two tile types $X$ and $Y$.
In detail, we select two integer parameters $c$ and $g$, 
and define the tile types $X : x \mapsto (x\star^g)^c x$ and $Y: y \mapsto (y\star^{g-1})^c y \star^c y$.
Both tile types have the same length $c(g+1) + 1$ and contain exactly $c+1$ numerals.
We then give $n$ tiles in the order $X,Y,X,Y,\ldots$ as input to our heuristics.
By construction, all defined sorting strategies behave the same since the order based on a stable sorting strategy is unchanged.
We observe the expected result in the left plot of \cref{fig:sameones}.
Only \texttt{random} is able to find a better placement.
To explain this phenomenon, we note that this chosen instance is favorable to random shuffling since the optimal solution would pick 
an order that groups $X$'s and $Y$'s individually.

In a second example, we reveal a bad behavior of \texttt{random}.
The idea is to introduce more freedom in the tile selection by adding a third tile type $Z$ defined as $Z = z \star^{c(g-2)} z (\star^{2g})^c z$.
This tile type has the same number of numerals and the same length as $X$ and $Y$.
If we additionally define filler tiles such that the minimal solution has no holes, we obtain the results shown in the right plot of \cref{fig:sameones}.
There, \texttt{random} performs increasingly worse than all other strategies for increasing $c$.

\clearpage

\bibliographystyle{plainnat}

\clearpage
\appendix

\section{Problems' Variants and Relationships}
\label{section:variants-relationships}

In this section, we analyze the possible variants of \probMaxShift, and their relationships and differences. In the setting of \probMaxShift{}, 
minimizing the length of the placement and minimizing the maximum size of the shift are equivalent: a solution has the maximum shift $\rho$ if and only if the corresponding placement has length $\ell + \rho$. First, we analyze the relationship between \probMaxShift and \probMaxShiftString (defined in~\ref{section:preliminaries}) :
\begin{remark}[\probMaxShift = \probMaxShiftString]
We note that problem \probMaxShift{} is equivalent to \probMaxShiftString{}.
Indeed, if $M$ is the input of \probMaxShift, we can consider $S_i$ as the $i$-th row of $M$.
Then, by definition of match, any placement $C$ given by \probMaxShift is such that each $S_i$ has a match in $C$, when seen as a partial string over the alphabet $\{\star\} \cup [1..n]$.
Furthermore, the amount $s(i)+1$ by which each row is shifted plus one is equal to the position at which it matches in the final placement.
Therefore, having maximum shift $\rho$ corresponds to the condition $\max_i s(i)+1 \leqslant \rho+1$.
Vice versa, any instance of \probMaxShiftString can be transformed into an instance of \probMaxShift by creating the corresponding binary matrix, with the same relationship between solutions. 
\end{remark}

\begin{example}
\label{example:smc-string}
Consider the instance $M$ of \probMaxShift given in Example~\ref{example:smc}.
The corresponding instance of \probMaxShiftString is given by the input strings $S_1 = \star{}1\star{}\star{}1\star{},$ $S_2 = \star{}2\star{}2\star{}\star{}$, $S_3 = \star{}3\star{}\star{}\star{}3$.
Indeed, the placement string $C=\star{}132123\star{}$ is also a solution of \probMaxShiftString{}, as we can easily see that $S_1$ has a match in $C$ at position 1, $S_2$ has a match in $C$ at position 3, and $S_3$ has a match in $C$ at position 2.
This corresponds to the shift function $s(1) = 0, s(2) = 2, s(3) = 1$, which provides a solution for \probMaxShift{}.
The matches can be represented as follows, with $C$ at the bottom:

$$\begin{matrix}
        \star{} & 1 & \star{} & \star{} & 1 & \star{}\\
         & & \star{} &2 & \star{} & 2 & \star{} & \star{}\\
         & \star{} &3 & \star{} & \star{} & \star{} & 3 \\
         \hline 
       \star{} &1 &3 &2 &1 &2 &3  &\star{} 
\end{matrix}$$
\end{example}

Looking at Examples~\ref{example:smc} and~\ref{example:smc-string}, a natural variation of \probMaxShift{} emerges: notice how the resulting placement $C$ has leading and trailing $\bsq{\star}$s, which could have been trimmed to obtain a shorter placement.
This leads to our variation,  \probMinLength{}, where we instead want to minimize the length of the placement $C$ \emph{by trimming the wildcard borders}. To this end, we require the input to be formed by trimmed binary strings (i.e. starting and ending with non-wildcard elements); formally defining \probMinLength as stated in the introduction.

\begin{remark}
    Requiring all input strings $S_i$ to be trimmed is equivalent to requiring the output minimum-length string $S$ to be trimmed. In other words, our current formulation of \probMinLength is equivalent to the following one:
    
\problembox{\textbf{Input:}
A set of $n$ binary strings $S_1, \ldots, S_n$ such that $S_i \in\{\star,i\}^{\ell}$ for some $\ell$ (i.e., all have the same length $\ell$ but distinct binary alphabets), and an integer $\rho \in [0..\ell n]$.
\\
\textbf{Task:}
Decide whether there exists a string $S'[1..\rho'] \in (\{\star\} \cup [1..n])^{\rho'}$ such that each $S_i$ has a match at some position of $S'$, and such that when trimming $S'$ we obtain $S$ of length $\rho$.
}

\end{remark}

This problem, which is the main focus of the paper, is quite different from the original \probMaxShift problem. Indeed, given the equivalence between minimizing the length and minimizing the maximum shift in the previous setting, one may think that solving \probMinLength on the trimmed input of \probMaxShift would be equivalent to trimming the output of \probMaxShift. This is in fact not the case:

\begin{remark}[$\probMinLength \neq \probMaxShift = \probMaxShiftString$]
Given a set of $n$ binary strings $S_1, \ldots, S_n$ such that $S_i \in\{\star,i\}^{\ell}$ for some $\ell$, solving \probMinLength on the trimmed versions $\hat S_1, \ldots, \hat S_n$ of $S_1, \ldots, S_n$ is different from trimming the solution of \probMaxShift (=\probMaxShiftString).

For instance, consider $S_1 = 11\star^h1$, and $S_2 = 2\star\star 2 \star^{h-1}$ with $h\geqslant 4$, both of length $\ell = h+3$. It is immediate to see that there is a solution to \probMaxShiftString for $\rho=1$, by shifting $S_1$ by one position to the right, yielding $P=2112\star^h1$ of length $\ell +1$. This solution is not trimmable, as it starts and ends with numerals. Let us now consider problem \probMinLength on the trimmed strings $\hat S_1 = S_1, \hat S_2 = 2\star \star 2$. In this case, the minimum-length solution is different: it has length $\ell$, and it is obtained by fitting $\hat{S_2}$ inside $\hat{S_1}$, for instance as $S= 112\star\star 2 \star^{h-4}1$. Note how in this solution the maximum shift is 2 instead of 1. 

\end{remark}

The previous remark showed something even stronger: in this new setting, it is no longer the case that minimizing the length of the placement and minimizing the maximum length of a shift are equivalent. 
This means that the following problem is distinct from \probMinLength{} as well:

\problembox{\probMaxShiftTiles{}
	\\
	\textbf{Input:}
A set of $n$ binary strings $S_1, \ldots, S_n$ such that $S_i \in i\{\star,i\}^{\ell_i} i$ for some $\ell_i$ (i.e., possibly different lengths and all over distinct binary alphabets),
and an integer $\rho \in [0..\ell n]$.
\\
\textbf{Task:}
Decide whether there exists a string $P$ with $P[i] \in \{\star\} \cup [1..n]$ for every $i \in [1..|P|]$, such that each $S_i$ has a match at some position $p_i \leqslant k$ of $P$. 
}

\begin{remark}
    \probMinLength is different from \probMaxShiftTiles. 
\end{remark}

Overall, we have four variant parameters, two concerning the input type (tiles or equal-length strings) and two concerning the output  objective function (minimizing the maximum shift or the placement length). The possible combinations boil down to three problems, summarized in the following table:
\begin{table}[h]
\begin{tabular}{ c | c | c | c }
\multicolumn{1}{c}{} & & \multicolumn{2}{c}{Input}\\
\cline{3-4}
\multicolumn{1}{c}{} & & Same-length Strings & Tiles \\ \hline
 
\multirow{2}{*}{Goal} & Minimizing the Max Shift & Problem \probMaxShift &  Problem \probMaxShiftTiles \\  
\cline{2-4}
 & Minimizing the Final Length & Problem \probMaxShift & Problem \probMinLength   
\end{tabular}
\end{table}

Reformulating \probMinLength{} as an optimization problem to find $\rho = \min_{s} \max_{i \in [1..n]} (s(i) + |S_i|)$, 
we observe that all problem variants are equivalent when all tiles have the same length.

In this work, we chose to focus on one of the two unexplored variants, namely \probMinLength. An in-depth study on the complexity of the last variant, \probMaxShiftTiles, is still open.

\section{Improved Algorithm for One Distinct Tile for Shorter Lengths}
The problem can also be solved in $O(\log n\cdot 16^\ell)$ time.
    Let $D[e][L][R]$ be the minimum length of a placement starting with $L$ and ending with $R$ that contains the tile $2^e$ times ($L$ and $R$ are bitmasks of length $\ell$, the minimum length includes $L$ but excludes $R$).

    Clearly, $D[e][L][R] = \min_{L',R'}\{D[e - 1][L][L'] + s_{L',R'} + D[e - 1][R'][R]\}$, where $s_{L',R'}$ is minimal s.t.\ $L'$ is disjoint from $s_{L',R'} + R'$. ($s_{L',R'}$ can be precomputed.)
    Compute $D[e][L][R]$ for all $e\leqslant \log_2{n}$ (this takes $O(\log n\cdot16^n)$ time).

    $D$ now gives us the answer if $n$ is a power of $2$. Let's generalize our solution to arbitrary $n$:
    We want to place $n$ copies of the tile, and the last $\ell$ positions of our placement should be $R$.
    Let $e'$ be maximal s.t.\ $2^{e'} \leqslant n$.
    Our answer is then \[
    \min_{L',R'}\{(\text{min.\ length of a pattern ending in $L'$ with $n-2^{e'}$ tiles}) + s_{L',R'} + D[e'][R'][R]\}
    \]
    In words, we compute the answer for $n$ tiles and ``the end'' $R$ from the minimum widths of a pattern for $n-2^{e'}$ tiles ending with all $2^\ell$ possible ends $L'$.
    
    We have $\log_2n$ recursion levels, each of which needs to return $2^\ell$ values, which in turn can be computed in $O((2^\ell)^2) = O(4^\ell)$ time from the results of the next lower recursion level, resulting in a time complexity of $O(\log n\cdot 8^\ell)$

\section{Missing Figures and Examples}

\begin{table*}[h]
	\centering
	\caption{Example of input set for the lower bound of \cref{section:approx-lowerbound}, for $\delta=4$. The placement produced by the leftmost-first greedy strategy would be $1\star^3 1\star^3 1\star^2 21\star 2\star 12 \star^2 2 3 \star 2 \star 32 \star^2 3\star^2 43 \star 4 \star 34 \star^2 45 \star 4 \star 54 \star^2 5 \star^2 65 \star 6 \star 5 6 \star^2 67 \star 6 \star 76 \star^2 7 \star^3 7 \star^3 7$ of length $80$, while an optimal solution is given by $(1357)^5 (246)^5$ of length $35$ having no holes (i.e., $\bsq{\star}$s). 
        }
		\label{table:example-lower-bound}
			\setlength{\tabcolsep}{1pt}
		\begin{tabular}{p{3em}*{22}{p{1.5em}}}
 $i$ & 1 & 2 & 3 & 4 & 5 & 6 & 7 & 8 & 9 & 10 & 11 & 12 & 13 &14 &15 &16 &17 
	\\
 $Y_1 = $  &
	$1$ & $\star{}$ & $\star{}$ & $\star{}$ & $1$ & $\star{}$ & $\star{}$ & $\star{}$ & $1$ & $\star{}$ & $\star{}$ & $\star{}$ & $1$ & $\star{}$ & $\star{}$ & $\star{}$ & $1$ 
	\\
	$X_1 = $ &
	$2$ & $\star{}$ & $\star{}$ & $2$ & $\star{}$ & $\star{}$ & $2$ & $\star{}$ & $\star{}$ & $2$ & $\star{}$ & $\star{}$ & $2$  & $\star{}$ & $\star{}$ & $2$  
     \\
     $Y_2 = $  &
	$3$ & $\star{}$ & $\star{}$ & $\star{}$ & $3$ & $\star{}$ & $\star{}$ & $\star{}$ & $3$ & $\star{}$ & $\star{}$ & $\star{}$ & $3$ & $\star{}$ & $\star{}$ & $\star{}$ & $3$ 
	\\
	$X_2 = $ &
	$4$ & $\star{}$ & $\star{}$ & $4$ & $\star{}$ & $\star{}$ & $4$ & $\star{}$ & $\star{}$ & $4$ & $\star{}$ & $\star{}$ & $4$  & $\star{}$ & $\star{}$ & $4$  
     \\
     $Y_3 = $  &
	$5$ & $\star{}$ & $\star{}$ & $\star{}$ & $5$ & $\star{}$ & $\star{}$ & $\star{}$ & $5$ & $\star{}$ & $\star{}$ & $\star{}$ & $5$ & $\star{}$ & $\star{}$ & $\star{}$ & $5$ 
	\\
	$X_3 = $ &
	$6$ & $\star{}$ & $\star{}$ & $6$ & $\star{}$ & $\star{}$ & $6$ & $\star{}$ & $\star{}$ & $6$ & $\star{}$ & $\star{}$ & $6$  & $\star{}$ & $\star{}$ & $6$  
	\\
     $Y_4 = $  &
	$7$ & $\star{}$ & $\star{}$ & $\star{}$ & $7$ & $\star{}$ & $\star{}$ & $\star{}$ & $7$ & $\star{}$ & $\star{}$ & $\star{}$ & $7$ & $\star{}$ & $\star{}$ & $\star{}$ & $7$ 
	\\
		\end{tabular}
	\end{table*}

\begin{algorithm}
\caption{DP-algorithm for \probMinLength{} with one distinct tile.}
\label{algminLengthDP}
\begin{algorithmic}[1]
\State \textbf{Function} minLengthDP($n$, $P$)
\State \textbf{Input:} $n$ tiles $P$ of length at most $\ell$
\State initialize $X[0][0] \gets 0$
\For{$i$ from $1$ to $n$}
	\For{each $Y$ with $X[i-1][Y]$ defined}
		\For{each distinct $\ell$-length suffix $Z$ among the strings in $Y \bowtie P$}
			\State Let $d$ be the number of new positions added by merging $Y$ and $P$
			\If{$X[i][Z]$ is undefined or $X[i-1][Y] + d < X[i][Z]$}
				\State $X[i][Z] \gets X[i-1][Y] + d$
			\EndIf
		\EndFor
	\EndFor
\EndFor
\State \Return the minimum defined value in $X[n][\cdot]$
\end{algorithmic}
\end{algorithm}

\clearpage
\section{Inapproximability}

We start with a general inapproximability result for both \probMinLength{} and \probMaxShift{}, and subsequently show a stronger result for \probMaxShift{}.

\begin{theorem}
	A $(2-\epsilon)$-approximation of \probMinLength{} and \probMaxShift{} is NP-hard for a constant $\epsilon \in (0,2]$.
\end{theorem}
\begin{proof}
	Given a multi-set of tiles $T_1,\ldots,T_n$ and an integer $\rho$ as an instance of \probMinLength{},
	we produce a new set of tiles $T'_i$ with $T'_i = T_i \star^\delta T_i$ for some $\delta \geqslant \ell$ to be defined later.
	We additionally create a tile $G = 1 \star^\rho 1^\delta \star^\rho 1$.
	Our new instance of \probMinLength{} is the tiles $T'_1,\ldots,T'_n$ and $G$ with the decision problem to fit them in a placement of length $2k + \delta + 2$.
	In case that the original instance is a YES instance, we can place all tiles $T'_i$ inside the placement initially filled with $G$.
	Otherwise, we need to shift $G$ or one of the tiles $T'_i$ outside the placement, which costs us at least $\delta$ additional positions.
Since $2(\ell+\delta)/ (2\ell + \delta ) \rightarrow_{\delta \rightarrow \infty} 2$, 
by scaling $\delta$ we can make the ratio arbitrary close to 2.
This means that a $(2-\epsilon)$ approximation is hard for any $\epsilon \in (0,2]$.

The same proof works for an instance of \probMaxShift{}, where $\rho$ then is the maximum required shift.
\end{proof}

In what follows, we show a stronger result for \probMaxShift{} by reducing from the chromatic number problem.
The \emph{chromatic number} of a graph $(V,E)$ is the least amount of colors needed to color the nodes of $V$ such that 
each node has a color different to all of its neighbors connected by $E$.
In what follows, we reduce hardness results for the approximation of the chromatic number to show hardness for the approximation of our problems.

\citet{lund94hardness} showed that approximating the chromatic number of a graph with $n$ vertices is NP-hard for a ratio $r = n^{\epsilon}$ for a constant $\epsilon > 0$
We can build on this result and show that \probMaxShift{} is also not $r$-approximable by reducing from $\rho$-coloring.
We extend the 3-coloring approach to a min-coloring approach by using $n$ colors as the number of columns assigned to an edge. 
Computing the minimum shift means that we compute the minimum amount of colors needed.
If there exists an approximation algorithm for \probMaxShift{}, it would directly give an approximation algorithm for the chromatic number as the objective functions have the same output, 
i.e, the chromatic number is the maximum needed shift, an L-reduction with constant one.
We conclude with the following theorem.

\begin{theorem}
	\probMaxShift{} cannot be approximated within a ratio of $r = n^{\epsilon}$ for a constant $\epsilon > 0$ unless P=NP.
\end{theorem}

\end{document}